\documentclass[twoside]{aiml16}

\usepackage{aiml16macro}

\usepackage{graphicx}
\usepackage{tikz}
\usepackage{amsmath,amssymb}
\usepackage{xcolor}
\usepackage{wasysym}
\usepackage{mathtools}
\usepackage{cleveref}
\usepackage{stmaryrd} 
\usepackage{enumerate}

\newcommand{\ineg}{{\smallsmile}}
\newcommand{\uneg}{{\smallfrown}}
\newcommand{\wsmile}{\scalebox{1}{$\mathrlap{\Circle}\textcolor{black}{\smallsmile}$}}
\newcommand{\wfrown}{\scalebox{1}{$\mathrlap{\Circle}\textcolor{black}{\smallfrown}$}}
\newcommand{\connec}{\pentagon}
\newcommand{\Circled}[1]{\textcircled{\scriptsize$#1$}}
\newcommand{\ttt}{t}
\newcommand{\fff}{f}
\newcommand{\TT}{\mathbf{T}}
\newcommand{\FF}{\mathbf{F}}
\newcommand{\XX}{\mathbf{X}}
\newcommand{\val}[4]{\ensuremath{{#1}^{#2}_{#3}(#4)}}
\newcommand{\tru}[2]{\val{\TT}{}{{\!#1}}{#2}}
\newcommand{\fal}[2]{\val{\FF}{}{{\!#1}}{#2}}
\newcommand{\truQ}[2]{\val{\TT}{Q}{#1}{#2}}
\newcommand{\falQ}[2]{\val{\FF}{Q}{#1}{#2}}
\newcommand{\xval}[2]{\val{\XX}{}{#1}{#2}}
\newcommand{\xvalQ}[2]{\val{\XX}{Q}{#1}{#2}}

\newcommand{\setP}{{\cal P}}
\newcommand{\setS}{{\cal L}}
\newcommand{\M}{{\cal M}}

\newcommand{\D}{{\bf D}}
\newcommand{\B}{{\bf B}}
\newcommand{\K}{{\bf K}}
\newcommand{\T}{{\bf T}}
\newcommand{\FUNC}{{\bf Fun}}
\newcommand{\Q}{{\cal Q}}
\newcommand{\F}{{\cal F}}
\newcommand{\E}{{\cal E}}

\newcommand{\wfr}{\prec}

\DeclareSymbolFont{symbolsC}{U}{txsyc}{m}{n}
\DeclareMathSymbol{\strictif}{\mathrel}{symbolsC}{74}
\newcommand{\suq}{\subseteq}
\newcommand{\set}[1]{\left\{#1\right\}}
\newcommand{\tup}[1]{\left<#1\right>}

\newcommand{\st}{{\ |\ }}   

\newcommand{\Ra}{\Rightarrow}

\newcommand{\rs}{{\,/\,}} 
\newcommand{\GKF}{{\rm {KF}}}
\newcommand{\GPK}{{\rm {PK}}}

\newcommand{\GPKB}{{\rm {PKB}}}
\newcommand{\GPKD}{{\rm {PKD}}}
\newcommand{\GPKT}{{\rm {PKT}}}
\newcommand{\GPKF}{{\rm {PKF}}}
\newcommand{\g}{\Gamma}
\renewcommand{\d}{\Delta}
\newcommand{\w}{\wedge}
\newcommand{\ssrul}[2]{\begin{array}{c}#1\\ \hline \ST #2\end{array}}
\newcommand\ST{\rule[-0.5em]{0pt}{1.5em}}

\usepackage{color}

\newcommand{\Lano}[2]{{#1}{{}\text{-}{#2}}}
\newcommand*{\ddfrac}[2]{\genfrac{}{}{0pt}{0}{#1}{#2}} 

\def\lastname{Lahav, Marcos, Zohar}

\begin{document}

\begin{frontmatter}
  \title{It ain't necessarily so: \\
  \ Basic sequent systems for negative modalities}
  \author{Ori Lahav}
\address{Max Planck Institute for Software Systems (MPI-SWS), Germany} 
  \author{Jo\~ao Marcos}
  \address{Federal University of Rio Grande do Norte, Brazil}
  \author{Yoni Zohar}
  \address{Tel Aviv University, Israel}
  
  \begin{abstract}
We look at non-classical negations and their corresponding adjustment connectives from a modal viewpoint, over complete distributive lattices, and apply a very general mechanism in order to offer adequate analytic proof systems to logics that are based on them. 
Defining non-classical negations within usual modal semantics automatically allows one to treat equivalent formulas as synonymous, and to have a natural justification for a global version of the contraposition rule. From that perspective, our study offers a particularly useful environment in which negative modalities and their companions may be used for dealing with inconsistency and indeterminacy. 
After investigating modal logics based on arbitrary frames, we extend the results to serial frames, reflexive frames, functional frames, and symmetric frames. In each case we also investigate when and how classical negation may thereby be defined.%
  \end{abstract}

 \end{frontmatter}


\section{Capturing the impossible, and its dual}
\label{sec:goals}
 Many well-known subclassical logics ---including intuitionistic logic and several many-valued logics--- share the conjunction-disjunction fragment of classical logic, but disagree about the exact notion of opposition and the specific logical features to be embodied in \textit{negation}.  
In contrast, modal logics are often thought of as superclassical, and are obtained by the addition of \textit{identity}-like `positive modalities'~$\Box$ and~$\lozenge$. 
For various well-known cases, such modalities fail to have a finite-valued characterization.
Notwithstanding, each $m$-ary connective~$\connec$ of a modal logic is typically \textit{congruential} (with respect to the underlying consequence relation~$\vdash$), in treating equivalent formulas as synonymous: 
if $\alpha_i\vdash\beta_i$ and $\beta_i\vdash\alpha_i$, for every $1\leq i\leq m$, then $\connec(\alpha_1,\ldots,\alpha_m)\vdash\connec(\beta_1,\ldots,\beta_m)$.
  To logical systems containing only such sort of connectives one might associate semantics in terms of neighborhood frames (see ch.5 of~\cite{Wojcicki88}), and the same applies if one uses 1-ary `negative modalities' instead, as in~\cite{ripley:PhD}. 
But normal modal logics make their 1-ary positive modalities respect indeed a stronger property: if $\alpha\vdash\beta$ then $\connec(\alpha)\vdash\connec(\beta)$. 
Such monotone behavior may be captured by semantics based on Kripke frames, and the same applies to the antitone behavior that characterize negative modalities, namely: if $\alpha\vdash\beta$ then $\connec(\beta)\vdash\connec(\alpha)$.

In~\cite{dun:zho:neggag:05} an investigation of negative modalities is accomplished on top of the ${\land}{\lor}{\top}{\bot}$-
fragment of classical logic, and the same base language had already been considered in~\cite{res:combpossneg:SL97} for the combination of positive and negative modalities.  Typically, in studies of positive and negative modalities the so-called compatibility (bi-relational) frames are used, and certain appropriate conditions upon the commutativity of diagrams involving their two relations are imposed, having as effect the heredity of truth (i.e., its persistence towards the future) with respect to one of the mentioned relations (assumed to be a partial order).

There are a number of studies (e.g.~\cite{vaka:cons89:full,Dos:NMOiIL:1984}) in which the above mentioned languages for dealing with negative modalities are upgraded in order to count on an (intuitionistic or classical) implication, and sometimes also its dual, co-implication (cf.~\cite{rau:BH:DM1980}).  If one may count on classical implication, however, it suffices to add to it the modal paraconsistent negation given by `unnecessity' (cf.~\cite{jmarcos:neNMLiP}), and all other connectives of normal modal logics turn out to be definable from such impoverished basis (indeed, where~$\ineg$ is a primitive symbol for unnecessity and $\to$ represents classical implication, we have that ${\sim}\alpha:=\alpha\to\ineg(\alpha\to\alpha)$ behaves as the classical negation of~$\alpha$, and $\Box\alpha:={\sim}\ineg\alpha$ behaves as the usual positive modality box).

Our intuition about the relation between a paracomplete (a.k.a.~`intui\-tion\-istic-like') negation and a paraconsistent negation is that the former would be expected to be more demanding than the latter, while classical negation should sit between the two (whenever it also turns out to be expressible). 
It takes indeed more effort to assert a negated statement constructively, while such statements are more readily asserted should some contradictions be allowed to subsist; in other words, negations in a paracomplete logic come at a greater cost than classical negations, while paraconsistent logics indulge on negations in which classical logic would show greater restraint.
The presence of a classical negation, however, often makes it too easy to forget that there are two distinct kinds of deviations equally worth studying, concerning non-classical negation, as one of these deviations may then be recovered in the standard way as the dual of the other.
In order to get a better grasp of the duality between paraconsistent and paracomplete modal negations (namely, unnecessity vs.\ impossibility), we purposefully make an effort to prevent the underlying language from being sufficiently expressive so as to allow for the definition of  a classical negation (or a classical implication) --- whenever that goal lies within reach.  Here we do however in all cases enrich our object language with certain `adjustment connectives' expressing negation-consistency and negation-determinacy, allowing for the simulation of usual features of classical negation and for the (partial) recovery of classical reasoning.  
It should be noted, however, that as a byproduct of the presence of such adjustment connectives truth will no longer be hereditary in our Kripke models, that is, it will not in general be preserved for all compound formulas towards the future, in contrast with what happens with models of compatibility frames.

In what follows, first and foremost we will concentrate on the logic $PK$, determined by the class of all Kripke frames, which has been introduced and received a presentation as a sequent system in~\cite{dod:mar:ENTCS2013}.  We show here that it can be reintroduced in terms of a so-called `basic sequent system', which allows one to take advantage of general techniques developed in~\cite{lah:avr:Unified2013}, including a method for obtaining sound and complete Kripke semantics and a uniform recipe for semantic proofs of cut-admissibility or analyticity.
The next section adopts a semantical perspective to explain why and how our study is done.%

\section{On negative modalities}
\label{sec:intro}
 We briefly recall the now familiar elements of a Kripke semantics.  
A~\textsl{frame} is a structure consisting of a nonempty set~$W$ (of `worlds') and a binary (`accessibility') relation~$R$ on~$W$. 
A \textsl{model} ${\cal M}=\langle {\cal F},V\rangle$ is based on a frame ${\cal F}=\langle W, R\rangle$ and on a \textsl{valuation} $V:W\times \setS \to \{\fff,\ttt\}$ that assigns truth-values 
to worlds $w\in W$ and sentences~$\varphi$ of a propositional language~$\setS$ generated over a denumerable set of propositional variables~$\setP$.
The valuations must satisfy certain conditions that are induced by the fixed interpretation of the connectives of the language.
When $V(w,\varphi)=\ttt$ we say that~$V$ \textsl{satisfies~$\varphi$ at~$w$}, and denote this by ${\cal M}, w\Vdash \varphi$; otherwise 
we write ${\cal M}, w\not\Vdash \varphi$ and say that~$V$ \textsl{leaves~$\varphi$ unsatisfied at~$w$}.
The connectives from the positive fragment of classical logic receive their standard boolean interpretations locally, world-wise, by recursively setting:

\smallskip

\noindent
\begin{tabular}{l l c l}
{[S$\top$]} & ${\cal M}, w\Vdash \top$\\
{[S$\land$]} & ${\cal M}, w\Vdash \varphi\land\psi$ & iff &  ${\cal M}, w\Vdash \varphi$ and ${\cal M}, w\Vdash\psi$\\
{[S$\lor$]} & ${\cal M}, w\not\Vdash \varphi\lor\psi$ & iff &  ${\cal M}, w\not\Vdash \varphi$ and ${\cal M}, w\not\Vdash\psi$\\
\end{tabular} 
\smallskip

\noindent
Given formulas $\Gamma\cup\Delta$ of~$\setS$, and given a class of frames~${\cal E}$, we say that \textsl{$\Gamma$ entails $\Delta$ in~${\cal E}$}, and denote this by $\Gamma\models_{\cal E}\Delta$, if for each model~${\cal M}$ based on a frame ${\cal F}\in{\cal E}$ and each world~$w$ of ${\cal M}$ we have either ${\cal M},w\not\Vdash\gamma$ for some $\gamma\in\Gamma$ or ${\cal M},w\Vdash\delta$ for some $\delta\in\Delta$.
The assertion $\Gamma\models_{\cal E}\Delta$ will be called a \textsl{consecution}.
As usual, in what follows we will focus most of the time on consecutions $\Gamma\models_{\cal E}\Delta$ involving a singleton~$\Delta$, and in the next section we will extend the notion of entailment so as to cover sequents instead of formulas. 
The subscript~${\cal E}$ shall be omitted in what follows whenever there is no risk of ambiguity.

In the following subsections we extend the above language with connectives whose modal interpretations will be useful for the investigation of non-classical negations.

\subsection{Adding negations}
\label{addingNEG}

Our first extension of the above language proceeds by the addition of a 1-ary connective~$\ineg$, to be interpreted  non-locally as follows:
\smallskip

\noindent
\begin{tabular}{l l c l}
{[S$\ineg$]} & ${\cal M}, w\Vdash \ineg\varphi$ & iff & ${\cal M}, v\not\Vdash \varphi$ for \underline{some} $v\in W$ such that $wRv$\\
\end{tabular} 
\smallskip

\noindent
Accordingly, a formula $\ineg\varphi$ is said to be satisfied at a given world of a model precisely when the formula $\varphi$ fails to be satisfied at some world accessible from this given world.  
In the following paragraph we will show that~$\ineg$ respects some minimal conditions to deserve being called a `negation', namely, we will demonstrate its ability to invert truth-values assigned to certain formulas (at certain worlds).

Let~$\#$ represent an arbitrary 1-ary connective, and let $\#^j$ abbreviate a $j$-long sequence of~$\#$'s.  The least we will demand from~$\#$ to call it a \textsl{negation} is that, for every $p\in\setP$ and every $k\in\mathbb{N}$:
\smallskip

\noindent
\begin{tabular}{l l l l}
{$\llbracket$\textit{falsificatio}$\rrbracket$} &  $\#^k p\not\models\#^{k+1} p$ \hspace{10mm} & 
{$\llbracket$\textit{verificatio}$\rrbracket$} &  $\#^{k+1} p\not\models\#^k p$\\
\end{tabular} 
\smallskip

\noindent 
To witness {$\llbracket$\textit{falsificatio}$\rrbracket$}, some sentence~$\varphi$ is to be satisfied while the sentence $\#\varphi$ is not simultaneously satisfied; for {$\llbracket$\textit{verificatio}$\rrbracket$} some sentence~$\varphi$ is left unsatisfied while at the same time $\#\varphi$ is satisfied. 
To check that the connective~$\ineg$ fulfills such requisites, it suffices for instance to build a frame in which $W=\{w_n:n\in\mathbb{N}\}$ and $wRv$ iff $v=w^{+\!+}$ (namely, $v$ is the successor of~$w$), and consider a valuation~$V$ such that $V(w_n,p)=\ttt$ iff $n$ is odd.

It is very easy to see that our connective~$\ineg$ 
satisfies \textsl{global contraposition}
in the sense that 
$\alpha\models\beta\mbox{ implies }\ineg\beta\models\ineg\alpha$.
Indeed, assume $\alpha\models\beta$ and suppose that  ${\cal M}, w\Vdash\ineg\beta$ for some world~$w$ of an arbitrary model~${\cal M}$.  Then, [S$\ineg$] informs us that there must be some world~$v$ in~${\cal M}$ such that $wRv$ and ${\cal M}, v\not\Vdash \beta$.  By the definition of entailment, the initial assumption gives us ${\cal M}, v\not\Vdash \alpha$.  Using again [S$\ineg$] we conclude that ${\cal M}, w\Vdash\ineg\alpha$.  As a byproduct of this, if one defines an equivalence relation~$\equiv$ on~$\setS$ by setting $\alpha\equiv\beta$ whenever both $\alpha\models\beta$ and $\beta\models\alpha$, then an easy structural induction on~$\setS$ establishes that~$\equiv$ is not only compatible with~$\ineg$ but also with the other connectives that are used in constructing the algebra of formulas; in other words, $\equiv$ constitutes a congruence relation on~$\setS$.  

It is straightforward to see that any 
1-ary connective~$\#$ satisfying global contraposition is such that, given $p,q\in\setP$:
\smallskip

\noindent
\begin{tabular}{l l l l}
(DM1.1\#) &  $\#(p\lor q)\models\# p\land \# q$ \hspace{5mm} & 
(DM2.1\#) &  $\# p\lor \# q\models\# (p\land q)$\\
\end{tabular} 
\smallskip

\noindent 
If~$\#$ also respects the following consecutions, then it is said to be a \textsl{full type 
diamond-minus connective}:
\smallskip

\noindent
\begin{tabular}{l l l l}
(DM2.2\#) &  $\#(p\land q)\models\# p\lor \# q$ \hspace{5mm} & %
(DT\#) &  $\# \top\models p$\\
\end{tabular} 
\smallskip

\noindent 
Note that~$\ineg$ is a full type 
diamond-minus connective.
To check that $\ineg$ satisfies (DM2.2\#), indeed, suppose that ${\cal M}, w\Vdash\ineg(p\land q)$ for some arbitrary world~$w$ of an arbitrary model~${\cal M}$.  By [S$\ineg$] we know that there is some world~$v$ such that $wRv$ and ${\cal M}, v\not\Vdash p\land q$.  It follows by [S$\land$] that ${\cal M}, v\not\Vdash p$ or  ${\cal M}, v\not\Vdash q$.  Using [S$\ineg$] again we conclude that ${\cal M}, w\Vdash\ineg p$ or ${\cal M}, w\Vdash\ineg q$ and [S$\lor$] gives us ${\cal M}, w\Vdash\ineg p\lor \ineg q$.  In addition, to check that~$\ineg$ satisfies (DT\#) one may invoke [S$\ineg$] and [S$\top$].  Note that satisfying (DT\#) means that the nullary connective~$\bot$ taken as an abbreviation of $\# \top$ is interpretable by setting, for every world~$w$ of every model~${\cal M}$:
\smallskip

\noindent
\begin{tabular}{l l c l}
{[S$\bot$]} & ${\cal M}, w\not\Vdash \bot$\\
\end{tabular} 
\smallskip

Given a negation~$\#$, we call the logic containing it \textsl{$\#$-paraconsistent} if the following consecution fails, for $p,q\in\setP$: 
\smallskip

\noindent
\begin{tabular}{l l}
{$\llbracket$\#-explosion$\rrbracket$} &  $p,\#p\models q$\\
\end{tabular} 
\smallskip

\noindent 
This means that there must be valuations that satisfy both some sentence~$\varphi$ and the sentence $\#\varphi$ while not satisfying every other sentence.  It is worth noticing that $\llbracket\ineg$-explosion$\rrbracket$ holds good in frames containing exclusively worlds that are accessible to themselves, and themselves only (call such worlds `narcissistic') and worlds that do not access any other world (call them `dead ends'): in the former case, it is impossible to simultaneously satisfy both~$\varphi$ and $\ineg\varphi$; in the latter case, the sentence $\ineg\varphi$ is never satisfied.  Note moreover that in the class of all narcissistic frames the connective~$\ineg$ happens to behave like classical negation, i.e., it behaves like the symbol~$\sim$ in the following semantic clause:
\smallskip

\noindent
\begin{tabular}{l l c l}
{[S$\sim$]} &  ${\cal M},w\Vdash {\sim} \varphi$ & iff & ${\cal M},w\not\Vdash \varphi$ \\
\end{tabular} 
\smallskip

\noindent 
In contrast, in the class of all frames whose worlds are all dead ends the connective~$\ineg$ does not respect [verificatio], and cannot be said thus to be a negation.

We now make a further extension of the above language by adding a 1-ary connective $\uneg$, non-locally interpreted as follows:
\smallskip

\noindent
\begin{tabular}{l l c l}
{[S$\uneg$]} & ${\cal M}, w\Vdash \uneg\varphi$ & iff & ${\cal M}, v\not\Vdash \varphi$ for \underline{\smash{every}} $v\in W$ such that $wRv$\\
\end{tabular} 
\smallskip

\noindent 
It is not difficult to check that again we have a connective that qualifies as a negation, and satisfies global contraposition.
\noindent
To reinforce the meta-theo\-re\-tical duality between the latter negation and the negation introduced above through [S$\ineg$], we will henceforth refer to the previous interpretation clause in the following equivalent form:
\smallskip

\noindent
\begin{tabular}{l l c l}
{[S$\uneg$]} & ${\cal M}, w\not\Vdash \uneg\varphi$ & iff & ${\cal M}, v\Vdash \varphi$ for \underline{some} $v\in W$ such that $wRv$\\
\end{tabular} 
\smallskip

A \textsl{full type 
box-minus connective} is a 1-ary connective~$\#$ that respects:
\smallskip

\noindent
\begin{tabular}{l l l l}
(DM1.2\#) &  $\# p\land \# q\models\#(p\lor q)$ \hspace{5mm} & %
(DF\#) &  $p\models \# \bot$\\
\end{tabular} 
\smallskip

\noindent 
One may easily check that~$\uneg$ is indeed a full type 
box-minus connective.

Given a negation~$\#$, we call the logic contaning it \textsl{$\#$-paracomplete} if it fails the following consecution, for $p,q\in\setP$: 
\smallskip

\noindent
\begin{tabular}{l l}
{$\llbracket$\#-implosion$\rrbracket$} &  $q\models \# p, p$\\
\end{tabular}
\smallskip

\noindent
Such failure will clearly be the case for~$\#=\uneg$ as soon as we entertain frames that contain worlds that are neither dead ends nor narcissistic.  Otherwise, we see that $\uneg$ will behave either like classical negation (if all worlds are narcissistic) or like~$\top$ (if all worlds are dead ends).

In the following sections, unless noted otherwise, we will no longer consider classes of frames containing only frames with worlds that are either dead ends or narcissistic --- so we will only consider entailment relations that are $\ineg$-paracon\-sist\-ent and $\uneg$-paracomplete, for the negative modalities $\ineg$ (assumed to be full-type diamond-minus) and~$\uneg$ (assumed to be full-type box-minus).

\subsection{Recovering negation-consistency and negation-determinacy}
\label{recovering}

In what follows we will call a model \textsl{dadaistic} when it contains some world in which all formulas are satisfied, and call it \textsl{nihilistic} if it leaves all formulas unsatisfied at some world.  It is straightforward to see that the language based on ${\land}{\lor}\top\ineg$, with the above interpretations, admits dadaistic models, while the language based on ${\land}{\lor}\bot\uneg$ admits nihilistic models.

Recall that a $\#$-paraconsistent logic allows for valuations that satisfy certain formulas~$\varphi$ and~$\# \varphi$ while leaving some other formula~$\psi$ unsatisfied (at some fixed world).  There might be reasons for disallowing this phenomenon to occur with an arbitrary~$\varphi$, or for restricting to certain formulas~$\psi$ but not others.  A particularly useful way of keeping a finer control over which `inconsistencies' of the form~$\varphi$ and~$\# \varphi$ are to be acceptable within non-dadaistic models is to mark down the formula thereby involved so as to recover a `gentle' version of $\llbracket$\#-explosion$\rrbracket$.  Concretely, for us here, a 1-ary connective~$\Circled{\#}$ that strongly internalizes the meta-theoretic `consistency assumption' at the object language level will be such that:
\smallskip

\noindent
\begin{tabular}{l l c l}
{[SC$\#$]} & ${\cal M}, w\Vdash \Circled{\#}\varphi$ & iff & ${\cal M}, w\not\Vdash \varphi$ or ${\cal M}, w\not\Vdash \#\varphi$\\
\end{tabular} 
\smallskip

\noindent
It is easy to check that any connective~$\Circled{\#}$ respecting [SC$\#$] is such that:
\smallskip

\noindent
\begin{tabular}{l l l l l l}
(C1\#) &  $\Circled{\#} p, p, \#  p\models$ \hspace{5mm} & 
(C2\#) &  $\models p,\Circled{\#} p$ \hspace{5mm} & 
(C3\#) &  $\models \# p,\Circled{\#} p$ \\
\end{tabular} 
\smallskip

\noindent
Note in particular that (C1\#) guarantees that there are no valuations that satisfy (at a fixed world) both $p$ and $\# p$ if these are put in the presence of $\Circled{\#} p$.  Thus, in case $\#$ fails $\llbracket$\#-explosion$\rrbracket$ we may look at the latter formula involving~$\Circled{\#}$ as guaranteeing that a weaker form of explosion is available.  On these grounds we shall call the connective~$\Circled{\#}$ an \textit{adjustment} companion to~$\#$: it allows one to recover explosion from within a non-\#-explosive (i.e., paraconsistent) logical context, and adjust the consecutions of the underlying logic so as to allow for the simulation of the consecutions that would otherwise be justified by reference to $\llbracket$\#-explosion$\rrbracket$.  Semantically, the presence of such connective also guarantees that dadaistic models are not admissible over the language based on ${\land}{\lor}\top\ineg\wsmile$, with the above interpretations.  This is because a formula of the form $\wsmile\varphi\land(\varphi\land\ineg\varphi)$ is equivalent to a formula~$\bot$ respecting [S$\bot$].

Dually, a \#-paracomplete logic allows for valuations that leave the formulas~$\varphi$ and~$\#\varphi$ both unsatisfied (at some fixed world), while satisfying some other formula~$\psi$. A particular way of keeping a finer control over which `indeterminacies' of the form $\varphi$ and~$\#\varphi$ are to be acceptable within non-nihilistic models is to allow for a `gentle' version of $\llbracket${\#-implosion}$\rrbracket$, where a 1-ary connective~$\Circled{\#}$ internalizes the meta-theoretic `determinacy assumption' at the object language level, in such a way that:
\smallskip

\noindent
\begin{tabular}{l l c l}
{[SD$\#$]} & ${\cal M}, w\not\Vdash \Circled{\#}\varphi$ & iff & ${\cal M}, w\Vdash \varphi$ or ${\cal M}, w\Vdash \#\varphi$\\ 
\end{tabular} 
\smallskip

\noindent
Clearly, any connective~$\Circled{\#}$ respecting [SD$\#$] is such that:
\smallskip

\noindent
\begin{tabular}{l l l l l l}
(D1\#) &  $\models\#  p, p, \Circled{\#} p$ \hspace{5mm} & 
(D2\#) &  $\Circled{\#} p,p\models $ \hspace{5mm} & 
(D3\#) &  $\Circled{\#} p,\# p\models$ \\
\end{tabular} 
\smallskip

Note that a formula of the form $(\varphi\lor\uneg\varphi)\lor\wsmile\varphi$ is equivalent to a formula~$\top$ respecting [S$\top$].
Note, moreover, that whenever it turns out that a connective~$\#$ respects $\llbracket${\#-explosion}$\rrbracket$ and at the same time its adjustment companion~$\Circled{\#}$ re\-spects [SC$\#$], then the formula $\Circled{\#}\varphi$ is equivalent to~$\top$.  In an analogous way, whenever a connective~$\#$ respects $\llbracket$\#-implosion$\rrbracket$ and at the same time its adjustment companion~$\Circled{\#}$ respects [SD$\#$], the formula $\Circled{\#}\varphi$ is equivalent to~$\bot$.  This stresses the fact that the adjustment connectives with which we deal in this subsection are more interesting when they accompany the respective non-classical negations to whose meaning they contribute.

At this point we have finally finished constructing the richest language that will be used throughout the rest of the paper: It will contain the connectives ${\land}{\lor}\top\bot\ineg\wsmile\uneg\wfrown$, disciplined by the [S\#] conditions above.  In the following subsection we will explain precisely when a classical negation, that is a 1-ary connective~$\sim$ subject to condition [S${\sim}$], is definable with the use of our language.  
Fixed such language, the logic characterized over it by the class~$\E$ of all frames will be called $PK$; the logic characterized by the class~$\E_{\D}$ of all frames with serial accessibility relations will be called $PKD$; the logic characterized by the class~$\E_{\T}$ of all frames with reflexive accessibility relations will be called $PKT$; the logic characterized by the class~$\E_{\FUNC}$ of all frames whose accessibility relations are total functions will be called $PKF$; the logic characterized by the class~$\E_{\B}$ of all symmetric frames (those with symmetric accessibility relations) will be called $PKB$.

\subsection{Around classical negation}
\label{aroundCN}

According to the intuitions laid down at \Cref{sec:goals}, one could expect that in general (a) $\uneg\alpha\vdash{\sim}\alpha$ and (b) ${\sim}\alpha\vdash\ineg\alpha$.  It is easy to see that these consecutions are sanctioned by $PKT$, for the classical negation~$\sim$ that may be defined by setting ${\sim}\varphi:=\ineg\varphi\land\wsmile\varphi$ (alternatively, one may set ${\sim}\varphi:=\uneg\varphi\lor\wfrown\varphi$).  

Meanwhile, in the deductively weaker logic $PKD$ one cannot in general prove (a) nor (b), even though a classical negation may be defined in this logic by setting ${\sim}\varphi:=(\uneg\varphi\land\wsmile\varphi)\lor\wfrown\varphi$.  However, one can still easily prove in $PKD$ that (c)~$\uneg\alpha\vdash\ineg\alpha$. 
In the logic $PKF$, deductively stronger than $PKD$ (but neither stronger nor weaker than $PKT$) one may also prove the converse consecution, (d) $\ineg\alpha\vdash\uneg\alpha$.  
Indeed, suppose ${\cal M}, w\Vdash\ineg\alpha$.  There is, by the fact that the accessibility relation is a total function, a single world~$v$ such that $wRv$.  Then ${\cal M}, v\not\Vdash\alpha$, by [S$\ineg$]. For a similar reason, invoking now [S$\uneg$] we conclude that ${\cal M}, w\Vdash\uneg\alpha$.  
Note that (c) and (d) together make our two modal non-classical negations indistinguishable from the viewpoint of $PKF$, yet there would still be no reason for them to collapse into classical negation.

The situation concerning classical negation and its relation to its non-clas\-sical neighbours gets even more interesting if one acknowledges that \textit{no} classical negation is definable in $PK$, the weakest of our logics, but also no classical ne\-gation is definable in the fragment of $PKT$ without neither of the adjustment connectives, or in the fragment of $PKF$ (or $PKD$) without either one of the adjustment connectives, or in $PKB$.  Detailed proofs concerning the mentioned results about (non)definability of classical negation in the weak modal logics that constitute our present object of study may be found in \Cref{sec:definability}. 

Notice that in $PKD$ and its extensions 
there are no negated formulas that happen to be true or false at a given world just because there are no worlds accessible from it.
Note also that the logic $PKT$ is: paraconsistent but not paracomplete with respect to the connective~$\ineg$; paracomplete but not paraconsistent with respect to~$\uneg$ 
(even though we will not prove it here, this logic is indeed the least extension of the positive implicationless fragment of classical logic with the latter mentioned properties). 
The logic $PKT$ will have its word in the following sections, for it also allows for the straightforward application of the techniques that will be hereby illustrated. 
In the other four mentioned logics, in contrast, both non-classical negations behave at once as paracomplete and paraconsistent negations (recall, though, that each is associated to a different adjustment connective).
We take the cases among these in which no classical negation is available to be particularly attractive for the task of revealing the `uncontamined' nature of non-classical negation.  Establishing well-behaved proof theoretical counterparts for such logics, as we shall do in what follows, is meant to allow for them to be better understood and dealt with. 


\section{A proof system for $PK$}
\label{sec:proofsystem}
 A sequent calculus for $PK$, that we denote by $\GPK$, was introduced in~\cite{dod:mar:ENTCS2013}, and consists of the following rules:
%
{\small
\[\begin{array}{ll@{\hspace{2em}}ll}
{[id]} & \ssrul{}{\g,\varphi\Ra \varphi,\d} &
 {[cut]} & \ssrul{\g,\varphi\Ra \d \ \ \ \ \g\Ra \varphi,\d}{\g\Ra \d}
\\
{[W{\Ra}]} & \ssrul{\g\Ra \d}{\g,\varphi\Ra \d}
& {[{\Ra}W]} & \ssrul{\g\Ra\d }{\g\Ra \varphi,\d} \\
{[{\bot}{\Ra}]} & \ssrul{}{\g,\bot\Ra \d} &
{[{\Ra}{\top}]} & \ssrul{}
{\g\Ra\top,\d} 
\\
{[{\w}{\Ra}]} & \ssrul{\g,\varphi,\psi\Ra \d}{\g,\varphi\w \psi\Ra \d} &
{[{\Ra}{\w}]} & \ssrul{\g\Ra \varphi,\d\ \ \ \g\Ra \psi,\d}
{\g\Ra \varphi\w \psi,\d} 
\\
{[{\vee}{\Ra}]} & \ssrul{\g,\varphi\Ra \d\ \ \ \g,\psi\Ra \d}
{\g,\varphi\vee \psi\Ra \d} & 
{[{\Ra}{\vee}]} & \ssrul{\g\Ra \varphi, \psi,\d}{\g\Ra \varphi\vee \psi,\d}
\\[2mm]
{[{\ineg}{\Ra}]} & \ssrul{\g\Ra\varphi, \d}
{\uneg\d,\ineg\varphi\Ra \ineg\g} & 
{[{\Ra}{\uneg}]} &
\ssrul{\g,\varphi\Ra \d}{\uneg\d\Ra \uneg\varphi,\ineg\g}
\\[2mm]
{[{\wsmile}{\Ra}]} & \ssrul{\g\Ra\varphi, \d\ \ \ \g\Ra\ineg\varphi,\d}
{\g,\wsmile\varphi\Ra\d} & 
{[{\Ra}{\wsmile}]} &
\ssrul{\g,\varphi,\ineg\varphi\Ra \d}{\g\Ra\wsmile\varphi,\d}
\\
{[{\wfrown}{\Ra}]} & \ssrul{\g\Ra\varphi,\uneg\varphi,\d}
{\g,\wfrown\varphi\Ra\d} & 
{[{\Ra}{\wfrown}]} &
\ssrul{\g,\varphi\Ra \d\ \ \ \g,\uneg\varphi\Ra\d}{\g\Ra\wfrown\varphi,\d}
\end{array}\]
}

\noindent
Above, sequents are taken to have the form $\Sigma\Ra\Pi$ where~$\Sigma$ and~$\Pi$ are finite sets of formulas, and given a unary connective~$\#$ and $\Psi\subseteq\mathcal{L}$, by $\#\Psi$ we denote the set $\set{\#\psi\st\psi\in\Psi}$.
We write $S\vdash_{\GPK}s$ to say that there is a derivation in $\GPK$ of a sequent~$s$ from a set~$S$ of sequents.
That establishes a consequence relation between sequents. 
A consequence relation between formulas is defined by setting 
$\g\vdash_{\GPK}\varphi$ if $\vdash_{\GPK}\g'\Ra\varphi$ for some finite subset~$\g'$ of~$\g$.  
The overloaded notation  $\vdash_{\GPK}$ will always be resolved by the pertinent context.

Next, we utilize in what follows the general mechanisms and methods applicable to the so-called `basic systems' of~\cite{lah:avr:Unified2013}
in order to prove soundness, completeness and cut-admissibility.
From the viewpoint of basic systems, each sequent is seen as a union of a `main sequent' and a `context sequent'. For example, in ${[{\Ra}{\vee}]}$, the main sequent of the premise is $\Ra\varphi,\psi$; the main sequent of the conclusion is $\Ra\varphi\vee\psi$; and the context sequent of both is $\g\Ra\d$. Note that in the rules for $\ineg$ and $\uneg$, the context sequent of the premise is different from the one of the conclusion. 
Accordingly, \cite{lah:avr:Unified2013} introduces the notion of a \textsl{basic rule}, whose premises
take the form $\tup{s,\pi}$, where~$s$ is a sequent that corresponds to the main sequent of the premise, and $\pi$ is a relation between singleton-sequents (that is, sequents of the form $\varphi\Ra$ or $\Ra\varphi$) called a \textsl{context relation} that determines the behavior of the context sequents. 
The sequent calculus $\GPK$ may be naturally regarded as a basic system that employs two context relations, namely:
$\pi_0=\set{\tup{p_1\Ra\;;\;p_1\Ra},\tup{\Ra p_1\;;\;\Ra p_1}}$, and 
$\pi_1=\set{ \tup{p_1\Ra\;;\;\Ra\ineg p_1}, \tup{\Ra p_1\;;\;\uneg p_1\Ra}}$. 
%
The rules of $\GPK$ may then be presented as particular instances of basic rules.
For example, the following are the basic rules for $\wedge,\ineg,\uneg$ and~$\wsmile$:
\smallskip

\noindent
\(\begin{array}{ll@{\hspace{1.5em}}ll}
{[{\Ra}{\wedge}]} & \tup{\Ra p_1;\pi_0},\tup{\Ra  p_2;\pi_0}\rs \Ra p_1\wedge p_2 &
{[{\wedge}{\Ra}]} & \tup{p_1, p_2\Ra;\pi_0}\rs p_1\wedge p_{2}\Ra \\
{[{\ineg}{\Ra}]} & \tup{\Ra p_1;\pi_1}\rs\ineg p_1\Ra &
{[{\Ra}{\uneg}]} & \tup{p_1\Ra ;\pi_1}\rs\Ra\uneg p_1 \\
{[{\wsmile}{\Ra}]} & \tup{\Ra p_1;\pi_0},\tup{\Ra \ineg p_1;\pi_0}\rs\wsmile p_1\Ra &
{[{\Ra}{\wsmile}]} & \tup{p_1,\ineg p_1\Ra;\pi_0}\rs\Ra\wsmile p_1
\end{array}\)
\smallskip

\noindent
In applications of $[{\Ra}{\wsmile}]$, the context sequent is left unchanged, as two singleton-sequents relate to each other (with respect to~$\pi_{0}$) iff they are the same. 
In contrast, applications of $[{\ineg}{\Ra}]$ are based on~$\pi_{1}$. A sequent $\g_{1}\Ra\d_{1}$ relates (with respect to~$\pi_{1}$) to a sequent $\g_{2}\Ra\d_{2}$ iff $\g_{2}=\uneg\d_{1}$ and $\d_{2}=\ineg\g_{1}$. 

We extend the notion of satisfaction from Section~\ref{sec:intro} to sequents by setting $\M,w\Vdash \g\Ra\d$ if $\M,w\not\Vdash\gamma$ for some $\gamma\in\g$ or $\M,w\Vdash\delta$ for some $\delta\in\d$.
Semantics for $\GPK$ may then be obtained using the general method introduced in~\cite{lah:avr:Unified2013}, by having 
each derivation rule and each context relation match a semantic condition, 
and the semantics of the system is obtained by conjoining all these semantic conditions.
For example, the basic rule $[{\ineg}{\Ra}]$ induces the condition:
``if $\M,v\Vdash\;\Ra\varphi$ for every world $v$ such that $w R v$, then $\M,w\Vdash\ineg\varphi\Ra$'',
which is equivalent to: ``If ${\cal M},w\Vdash \ineg\varphi$  then  ${\cal M}, v\not\Vdash \varphi$ for some $v\in W$ such that $wRv$''.  This is half of clause [S$\ineg$], from Section~\ref{sec:intro}.
Furthermore, the context relation $\pi_{1}$ induces an additional semantic condition: ``if $w R v$ then $\M,w\Vdash\; \Ra\ineg\varphi$ whenever $\M,v\Vdash\varphi\Ra$''. 
This amounts to the other half of clause [S$\ineg$], namely:
``${\cal M},w\Vdash \ineg\varphi$  whenever  ${\cal M}, v\not\Vdash \varphi$ for some $v\in W$ such that $wRv$''. 
Systematically applying this semantic reading to all rules and all context relations of $\GPK$ (according to Definitions 4.5 and 4.12 of \cite{lah:avr:Unified2013}), one obtains soundness and completeness 
with respect to the class of all Kripke models $\tup{\F,V}$, where~$\F$ is an arbitrary frame and each valuation $V:W\times\setS\rightarrow\set{\fff,\ttt}$ respects the following conditions, for every $w\in W$ and $\varphi,\psi\in\setS$:

\smallskip
\noindent
\begin{tabular}{rl}
	{[$\TT\top$]} & $\tru{w}{\top}$\\[.5mm]
	{[$\FF\bot$]} & $\fal{w}{\bot}$\\[.5mm]
\end{tabular} \\
\begin{tabular}{rl}
	{[$\TT\land$]} & if $\tru{w}{\varphi}$ and $\tru{w}{\psi}$, then $\tru{w}{\varphi\w\psi}$\\ 
	{[$\FF\land$]} & if $\fal{w}{\varphi}$ or $\fal{w}{\psi}$, then $\fal{w}{\varphi\w\psi}$
	\\[.5mm] 
\end{tabular} \\
\begin{tabular}{rl}
	{[$\TT\lor$]} & if $\tru{w}{\varphi}$ or $\tru{w}{\psi}$, then $\tru{w}{\varphi\vee\psi}$\\ 
	{[$\FF\lor$]} & if $\fal{w}{\varphi}$ and $\fal{w}{\psi}$, then $\fal{w}{\varphi\vee\psi}$
	\\[.5mm] 
\end{tabular} \\
\begin{tabular}{rl}
	{[$\TT\ineg$]} & if $\fal{v}{\varphi}$ for some $v\in W$ such that $wRv$, then $\tru{w}{\ineg\varphi}$\\ 
	{[$\FF\ineg$]} & if $\tru{v}{\varphi}$ for every $v\in W$ such that $wRv$, then $\fal{w}{\ineg\varphi}$
	\\[.5mm] 
\end{tabular} \\
\begin{tabular}{rl}
	{[$\TT\uneg$]} & if $\fal{v}{\varphi}$ for every $v\in W$ such that $wRv$, then $\tru{w}{\uneg\varphi}$\\
	{[$\FF\uneg$]} & if $\tru{v}{\varphi}$ for some $v\in W$ such that $wRv$, then $\fal{w}{\uneg\varphi}$
	\\[.5mm] 
\end{tabular} \\
\begin{tabular}{rl}
	{[$\TT\wsmile$]} & if $\fal{w}{\varphi}$ or $\fal{w}{\ineg\varphi}$, then $\tru{w}{\wsmile\varphi}$\\
	{[$\FF\wsmile$]} & if $\tru{w}{\varphi}$ and $\tru{w}{\ineg\varphi}$, then $\fal{w}{\wsmile\varphi}$
	\\[.5mm]
\end{tabular} \\
\begin{tabular}{rl}
	{[$\TT\wfrown$]} & if $\fal{w}{\varphi}$ and $\fal{w}{\uneg\varphi}$, then $\tru{w}{\wfrown\varphi}$\\
	{[$\FF\wfrown$]} & if $\tru{w}{\varphi}$ or $\tru{w}{\uneg\varphi}$, then $\fal{w}{\wfrown\varphi}$\\
\end{tabular} 
\medskip

\noindent 
where we take `$\tru{u}{\alpha}$' as abbreviating `$V(u,\alpha)=\ttt$', and `$\fal{u}{\alpha}$' as abbreviating `$V(u,\alpha)=\fff$'.  If alternatively one just \textit{rewrites} $V(v,\alpha)=\ttt$ as ${\cal M}, v\Vdash \alpha$ and rewrites $V(v,\alpha)=\fff$ as ${\cal M}, v\not\Vdash \alpha$, where ${\cal M}=\tup{\tup{W,R},V}$, what results thereby is a collection of conditions that are essentially identical to the [S\#] clauses introduced in our Section~\ref{sec:intro}.

Two brief comments are in order here.
First, our valuation functions 
assign truth-values to \emph{every} formula in every world.  However, as the values of compound formulas are uniquely determined by the values of their subformulas, we could have rested content above with assigning truth-values to propositional variables.
Second, given that for the above valuations $\tru{u}{\alpha}$ is the case iff $\fal{u}{\alpha}$ fails to be the case, the semantic conditions [$\TT\#$] and [$\FF\#$], for each connective~$\#$, are clearly the converse of each other.  
In setting the two conditions apart, 
we have just given them directionality, pointing from less complex to more complex formulas, and have separated between conditions induced by rules from those induced by context relations.
While neither of these manoeuvres are very useful here, they will allow us to more easily relate, in Section~\ref{sec:analyticity}, valuations to `quasi valuations' that have non truth-functional semantics.

Fix in what follows a Kripke model ${\cal M}=\tup{\tup{W,R},V}$.  
We say that $w,v\in W$ \textsl{agree with respect to the formula~$\alpha$, according to~$V$}, if 
either ($\tru{w}{\alpha}$ and $\tru{v}{\alpha}$) or ($\fal{w}{\alpha}$ and $\fal{v}{\alpha}$).
We say that ${\cal M}$ is \textsl{differentiated} if we have $w=v$ whenever~$w$ and~$v$ agree with respect to every $\alpha\in{\cal L}$, according to~$V$.
We call ${\cal M}$ a \textsl{strengthened} model if $w R v$ iff 
($\tru{v}{\alpha}$ implies $\fal{w}{\uneg\alpha}$) and ($\fal{v}{\alpha}$ implies $\tru{w}{\ineg\alpha}$), for every $\alpha\in{\cal L}$.  It is worth stressing that the accessibility relation of a strengthened model is uniquely determined by the underlying collection of worlds and valuation.
The following result follows directly from Corollary~4.26 in~\cite{lah:avr:Unified2013}, thus there is no need to prove it again here:

\begin{theorem}\label{strongerSoundnessAndCompletenessForSmiles}
$\GPK$ is sound and complete with respect to any class of Kripke models that: 
$(i)$ contains only models that satisfy all the above $[\TT\#]$ and $[\FF\#]$ conditions;
and $(ii)$ contains all strengthened differentiated models that satisfy all the above $[\TT\#]$ and $[\FF\#]$ conditions.
\end{theorem}

This theorem provides a mechanism that will be recycled in the subsequent sections, when we consider extensions of $\GPK$.
The following result from \cite{dod:mar:ENTCS2013} comes as a byproduct of it:

\begin{corollary}\label{soundnessAndCompletenessForSmiles}
$\g\models_{\E}\varphi$ iff $\g\vdash_{\GPK}\varphi$ for every $\g\cup\set{\varphi}\suq\setS$, where ${\cal E}$ denotes the class of all frames.
\end{corollary}

\section{(Almost) Free Lunch: cut-elimination and analyticity}
\label{sec:analyticity}
 In this section we make further use of the powerful machinery introduced in~\cite{lah:avr:Unified2013} to prove that PK enjoys strong cut-admissibility, 
in other words, we show that $S\vdash_{\GPK}s$ implies that there is a derivation in $\GPK$ of the sequent~$s$ from the set of sequents~$S$ such that in every application of the cut rule the cut formula~$\varphi$ appears in~$S$.
In particular, $\vdash_{\GPK}s$ implies that~$s$ is derivable in $\GPK$ without any use of the cut rule.
The proof is done in two steps. \textit{First}, we present an adequate 
semantics for the cut-free fragment of $\GPK$.
\textit{Second}, we show that a countermodel in this new semantics entails the existence of a countermodel in the form of a Kripke model as defined in the previous section. This, together with Corollary~\ref{soundnessAndCompletenessForSmiles}, entails that $\GPK$ is equivalent to its cut-free fragment.

\subsection*{Step 1. Semantics for cut-free $\GPK$}

Semantics for cut-free basic systems may be obtained 
through the use of `quasi valuations'. 
Models based on quasi valuations differ from usual Kripke models in two main aspects: 
$(a)$ the underlying interpretation is three-valued; 
$(b)$ the underlying interpretation is non-deterministic --- the truth-value of a compound formula in a given world is \textit{not} always uniquely determined by the truth values of its subformulas in the collection of worlds of the underlying frame.

To obtain such semantics for $\GPK$, as before, one reads off a semantic condition on quasi valuations from each derivation rule and from each context relation.
As per Theorems 5.24 and 5.31 of \cite{lah:avr:Unified2013}, we know that 
the class of models based on quasi valuations that respect all these conditions is sound and complete for the cut-free fragment of $\GPK$. 
Concretely, 
given a frame $\F=\tup{W,R}$, a \textsl{quasi valuation} over it is a function $QV:W\times\setS\rightarrow \set{\set{\fff},\set{\ttt},\set{\fff,\ttt}}$ satisfying precisely the same semantic conditions laid down in Section~\ref{sec:proofsystem}, where we now take `$\tru{u}{\alpha}$' as abbreviating `$\ttt\in V(u,\alpha)$', and `$\fal{u}{\alpha}$' as abbreviating `$\fff\in V(u,\alpha)$'.  
Whenever we need to distinguish between a semantic condition on a tuple $\langle w,\varphi\rangle$ as constraining a valuation~$V$ or a quasi valuation $QV$, we will use $\xval{w}{\varphi}$ for the former and  $\xvalQ{w}{\varphi}$ for the latter, where $\XX\in\{\TT,\FF\}$.
A \textsl{quasi model} is a structure $\Q\M=\tup{\F,QV}$, where $QV$ is a quasi valuation over~$\F$.  
The notions of a differentiated quasi model and of a strengthened quasi model are defined as before, assuming the same abbreviations.

\subsection*{Step 2. Semantic cut-admissibility}

The next step is to show that the existence of a countermodel in the form of a strengthened differentiated quasi model implies the existence a countermodel in the form of an ordinary Kripke model (following Corollary 5.48 of \cite{lah:avr:Unified2013}). 
For this purpose we 
define 
an \textsl{instance} of a quasi model $\Q\M=\tup{\tup{W,R},QV}$ as any model of the form $\M=\tup{\tup{W,R'},V}$ such that 
$\xvalQ{w}{\varphi}$ whenever $\xval{w}{\varphi}$, for every $\XX\in\{\TT,\FF\}$,
every $w\in W$ and every $\varphi\in{\cal L}$.  Note that a quasi model and its instances may have different accessibility relations.

In what follows, the construction of appropriate instances is done by a recursive definition over the following well-founded relation~$\wfr$ on the set of formulas: 
%
$\alpha\wfr\beta$ if either
$(i)$  $\alpha$ is a proper subformula of $\beta$;
$(ii)$ $\alpha=\ineg\gamma$ and $\beta=\wsmile\gamma$ for some $\gamma\in{\cal L}$;
or $(iii)$ $\alpha=\uneg\gamma$ and $\beta=\wfrown\gamma$ for some $\gamma\in{\cal L}$.
%

\begin{lemma}
\label{quasinstance}
	Every quasi model has an instance.
\end{lemma}
\begin{proof}
	Let $\Q\M=\tup{\F,QV}$ be a quasi model based on a frame $\F=\tup{W,R}$. 
We set us now an appropriate valuation $V:W\times\setS\rightarrow\set{\fff,\ttt}$.
For every world~$w$ and formula~$\varphi$, the valuation~$V$ is inductively defined (with respect to~$\wfr$) on~$\varphi$ as follows: 
(R1) if $\truQ{w}{\varphi}$ fails for $QV$, we postulate $\fal{w}{\varphi}$ to be the case for~$V$; 
(R2) if $\falQ{w}{\varphi}$ fails for $QV$, we postulate $\tru{w}{\varphi}$ to be the case for~$V$; 
(R3) otherwise both $\truQ{w}{\varphi}$ and $\falQ{w}{\varphi}$ hold good for $QV$, and in this case we postulate $\tru{w}{\varphi}$ to be the case for~$V$ if one of the following holds:
\smallskip

\noindent
\begin{tabular}{rl}
  {(M1)} & $\varphi$ is a propositional variable or $\varphi$ is $\top$\\
  {(M2)} & $\varphi=\varphi_{1}\w\varphi_{2}$, and both $\tru{w}{\varphi_{1}}$ and $\tru{w}{\varphi_{2}}$\\
  {(M3)} & $\varphi=\varphi_{1}\vee\varphi_{2}$, and either $\tru{w}{\varphi_{1}}$ or $\tru{w}{\varphi_{2}}$\\
  {(M4)} & $\varphi=\ineg\psi$, and $\fal{v}{\psi}$ for some $v\in W$ such that $w R v$\\
  {(M5)} & $\varphi=\uneg\psi$, and $\fal{v}{\psi}$ for every $v\in W$ such that $w R v$\\
  {(M6)} & $\varphi=\wsmile\psi$, and either $\fal{w}{\psi}$ or $\fal{w}{\ineg\psi}$\\
  {(M7)} & $\varphi=\wfrown\psi$, and both $\fal{w}{\psi}$ and $\fal{w}{\uneg\psi}$\\
\end{tabular}\smallskip
	
\noindent 
Otherwise, we postulate $\fal{w}{\varphi}$ to be the case for~$V$.
Obviously, $\xval{w}{\varphi}$ implies $\xvalQ{w}{\varphi}$ for every $w\in W$, every $\varphi\in{\cal L}$ and every $\XX\in\{\TT,\FF\}$.
It is routine to verify that $\tup{\F,V}$ is a model. 
We show here that the semantic conditions for $\ineg$ and $\wsmile$ hold:\\
\noindent
\textbf{[Case of $\ineg$]} 
Let $\psi\in{\cal L}$.  
Suppose first that $\fal{v}{\psi}$ is the case for some~$v\in W$ such that $wRv$.  Then $\falQ{v}{\psi}$.  Since $\Q\M$ is a quasi model, then $\truQ{w}{\ineg\psi}$ is the case.  If, on the one hand, $\falQ{w}{\ineg\psi}$ fails, then we must have $\tru{w}{\ineg\psi}$, by (R2).  If, on the other hand, neither $\truQ{w}{\ineg\psi}$ nor $\falQ{w}{\ineg\psi}$ fail, we are in case (R3).  Since we have $\fal{v}{\psi}$ and $wRv$ we conclude by (M4) that $\tru{w}{\ineg\psi}$ must be the case.  
Suppose now that $\tru{v}{\psi}$ is the case for every world~$v$ such that $wRv$.  Then we have $\truQ{v}{\psi}$ for every such world.  Since $\Q\M$ is a quasi model, it follows that $\falQ{w}{\ineg\psi}$ is the case.  If, on the one hand, $\truQ{w}{\ineg\psi}$ fails, then we must have $\fal{w}{\ineg\psi}$, by (R1). If, on the other hand, neither $\truQ{w}{\ineg\psi}$ nor $\falQ{w}{\ineg\psi}$ fail, we are in case (R3).  Since we have $\tru{v}{\psi}$ for every world $v$ such that $wRv$ we conclude that none of (M1)--(M7) applies, thus $\fal{w}{\ineg\psi}$ must be the case.  \\
\textbf{[Case of }$\wsmile$\textbf{]}
Let $\psi\in\setS$. 
Suppose first that either $\fal{w}{\psi}$ or $\fal{w}{\ineg\psi}$ are the case for some $w\in W$. Then either $\falQ{w}{\psi}$ or $\falQ{w}{\ineg\psi}$. Since $\Q\M$ is a quasi model, it follows that $\truQ{w}{\wsmile\psi}$. If, on the one hand, $\falQ{w}{\wsmile\psi}$ fails, then we must have $\tru{w}{\wsmile\psi}$, by (R2). If, on the other hand, neither $\truQ{w}{\wsmile\psi}$ nor $\falQ{w}{\wsmile\psi}$ fail, we are in case (R3) and we conclude by (M6) that $\tru{w}{\wsmile\psi}$ must be the case.
Suppose now that both $\tru{w}{\psi}$ and $\tru{w}{\ineg\psi}$ are the case for some $w\in W$. Then $\truQ{w}{\psi}$ and $\truQ{w}{\ineg\psi}$. Since $\Q\M$ is a quasi model, then $\falQ{w}{\wsmile\psi}$. If, on the one hand, $\truQ{w}{\wsmile\psi}$ fails, then we must have $\fal{w}{\wsmile\psi}$, by (R1). If, on the other hand, neither $\truQ{w}{\wsmile\psi}$ nor $\falQ{w}{\wsmile\psi}$ fail, 
we are in case (R3) and $\fal{w}{\wsmile\psi}$ must be the case because none of (M1)--(M7) applies.
\end{proof}

Since the class of all quasi models contains the strengthened differentiated quasi models, it follows that:

\begin{corollary}
\label{gpkcut}
	$\GPK$ enjoys strong cut-admissibility. 
\end{corollary}

\begin{corollary}\label{PKanalyticity}
$\GPK$ is $\wfr$-analytic: If a sequent $s$ is derivable from a set $S$ of sequents in $\GPK$, then there is a derivation of $s$ from $S$ 
such that every formula $\varphi$ that occurs in the derivation satisfies $\varphi\wfr\psi$ for some $\psi$ in $S\cup s$.
\end{corollary}
\begin{proof}
By induction on the length of the derivation of $s$ from $S$ in $\GPK$:
In all rules except for $(cut)$, the premises include only formulas~$\varphi$ that satisfy $\varphi\wfr\psi$ for some formula~$\psi$ in the conclusion. 
\end{proof}

\section{Some special classes of frames}
\label{sec:extensions}
 In this section we present three very natural deductive extensions of $\GPK$.
Given a property~$X$ of binary relations, we call a frame $\tup{W,R}$ an {\em $X$~frame} if~$R$ enjoys~$X$. A (quasi) model $\tup{\F,V}$ is called an {\em $X$~$($quasi$)$ model} if $\F$ is an~$X$~frame. In addition, and similarly to what we did in the case of $\GPK$, for every proof system $Y$ we write $S\vdash_{Y}s$ if there is a derivation of $s$ from $S$ in $Y$.

\subsection{Seriality}

Let $\GPKD$ be the system obtained
by augmenting $\GPK$ with the following rule: 
{\small
$$[\D] \quad \dfrac{\Gamma\Ra\Delta}{\uneg\Delta\Ra\ineg\Gamma}$$
}%
This rule may be formulated as the basic rule:
$\tup{\Ra\;;\;\pi_{1}}\rs\Ra$. 
Since its premise  is the empty sequent, the semantic condition it imposes (following~\cite{lah:avr:Unified2013}) is seriality: indeed, respecting $[\D]$ in a world $w$ of a model $\M$ based on a frame $\tup{W,R}$ means that if $\M,v\Vdash\; \Ra$ for every world $v$ such that $w R v$, then also $\M,w\Vdash\; \Ra$. Since the empty sequent is not satisfied at any world, this condition would hold iff for every world $w$ there exists a world $v$ such that $w R v$.  A~similar argument shows that every serial frame satisfies this semantic condition.

As in Corollary \ref{soundnessAndCompletenessForSmiles}, we obtain a completeness theorem for $\GPKD$ with respect to serial models:
\begin{corollary}\label{soundnessAndCompletenessForSmilesatKD}
$\g\models_{\E_{\D}}\varphi$ iff $\g\vdash_{\GPKD}\varphi$ for every $\g\cup\set{\varphi}\suq\setS$, where $\E_{\D}$ is the class of serial models. 
\end{corollary}

Additionally, we may prove cut-admissibility also for $\GPKD$, going through serial quasi models.

\begin{lemma}
\label{Dinstance}
	Every  serial quasi model has a serial instance.
\end{lemma}
\begin{proof}
The proof is the same as the proof of Lemma \ref{quasinstance}.
Note indeed that no property of the accessibility relation was assumed, and the constructed instance 
has the same accessibility relation as the original quasi model.
\end{proof}


\begin{corollary}
	$\GPKD$ enjoys cut-admissibility and is $\wfr$-analytic.
\end{corollary}

\subsection{Reflexivity}

Let $\GPKT$ be the system obtained
by augmenting $\GPK$ with the following rules:
{\small
$$[{\Ra}\ineg] \quad \dfrac{\Gamma,\varphi\Ra\Delta}{\Gamma\Ra\ineg\varphi,\Delta}
\qquad\qquad\qquad [\uneg{\Ra}] \quad \dfrac{\Gamma\Ra\varphi,\Delta}{\Gamma,\uneg\varphi\Ra\Delta}$$}%
These rules may be formulated as the basic rules:
$\tup{p_{1}\Ra\;;\;\pi_{0}}\rs\Ra\ineg p_{1}$ and $\tup{\Ra p_{1}\;;\;\pi_{0}}\rs\uneg p_{1}\Ra\;$.  It should be clear that $\GPKT$ allows thus for the derivation of the consecutions representing $\llbracket\ineg$-implosion$\rrbracket$ and $\llbracket\uneg$-explosion$\rrbracket$.

Semantically, they impose reflexivity not on all models, but only on {\em strengthened} models. Indeed, since the underlying context relation is $\pi_{0}$, for every  model $\M=\tup{\F,V}$ based on a frame $\F=\tup{W,R}$ that respects $[{\Ra}\ineg]$ and $[\uneg{\Ra}]$,  and every world $w$, if $\M,w\vDash\varphi\Ra$  then  $\M,w\vDash\;\Ra\ineg\varphi$ and if $\M,w\vDash\;\Ra\varphi$ then  $\M,w\vDash\uneg\varphi\Ra$.
To put it otherwise, 
if $\fal{w}{\varphi}$ then $\tru{w}{\ineg\varphi}$, and
if $\tru{w}{\varphi}$ then $\fal{w}{\uneg\varphi}$.
Clearly, every reflexive model satisfies these conditions. To show that every strengthened model that satisfies them is reflexive, consider an arbitrary strengthened model $\M=\tup{\tup{W,R},V}$. Then for every world $w\in W$ we have that for every formula $\varphi$, ($\tru{w}{\varphi}$ implies $\fal{w}{\uneg\varphi}$)
and ($\fal{w}{\varphi}$ implies $\tru{w}{\ineg\varphi}$), which in strengthened models means precisely that $w R w$.
We obtain thus a completeness theorem for $\GPKT$ with respect to reflexive models:
\begin{corollary}\label{soundnessAndCompletenessForSmilesatKT}
$\g\models_{\E_{\T}}\varphi$ iff $\g\vdash_{\GPKT}\varphi$ for every $\g\cup\set{\varphi}\suq\setS$, where $\E_{\T}$ is the class of reflexive models.
\end{corollary}

Such semantics for $\GPKT$ allows one to easily confirm that the full type diamond-minus connective~$\ineg$ fails (DM1.2\#), and that the full type box-minus connective~$\uneg$ fails (DM2.2\#).   These properties transfer to the weaker logics $\GPKD$ and $\GPK$, of course.

Cut-admissibility for $\GPKT$ may be obtained using arguments similar to those used in proving Lemma~\ref{quasinstance}.
It follows thus that:

\begin{lemma}
	Every  reflexive strengthened quasi model has a reflexive instance.
\end{lemma}

\begin{corollary}
	$\GPKT$ enjoys cut-admissibility and is $\wfr$-analytic.
\end{corollary}

\subsection{Functionality}
In this section we address functional frames, that is, frames whose accessibility relations are {\em  total functions}.
In every model $\tup{\tup{W,R},V}$ of a functional frame and world $w\in W$, we have 
$\tru{w}{\ineg\varphi}$ iff $\tru{w}{\uneg\varphi}$. Hence $\ineg$ and $\uneg$ are indistinguishable. 
Accordingly, here we consider a restricted language, without~$\uneg$.

Let $\GPKF$ be the system obtained from $\GPK$ by substituting~$\ineg$ for~$\uneg$ in rules $[{\Ra}\wfrown]$ and $[\wfrown{\Ra}]$, and replacing both rules $[{\Ra}\uneg]$ and $[\ineg{\Ra}]$ with the single rule:
{\small
$$[\FUNC]\quad\dfrac{\Gamma\Ra\Delta}{\ineg\Delta\Ra\ineg\Gamma}$$
}%
It is straightforward to see that rule $[\FUNC]$ may be formulated as the following basic rule:
$\tup{\Ra\;;\;\pi_{2}}\rs\Ra$, for $\pi_{2}=\set{\tup{\varphi\Ra\;;\;\Ra\ineg\varphi},\tup{\Ra\varphi\;;\;\ineg\varphi\Ra}}$.

The latter rule and context relation impose functionality on {\em differentiated} models. Indeed, respecting the basic rule $[\FUNC]$ corresponds to seriality, similarly to the case of the rule $[\D]$.
Additionally, the context relation $\pi_{2}$ forces the accessibility relation to be a partial function:
respecting $\pi_{2}$ in a world $w$ of a model $\M=\tup{\tup{W,R},V}$ means that for every $v_{1},v_{2}\in W$ such that $w R v_{1}$ and $w R v_{2}$ and for every formula $\varphi$ we have that 
$\tru{v_{1}}{\varphi}$ iff $\fal{w}{\ineg\varphi}$ iff 
$\tru{v_{2}}{\varphi}$.  When $\M$ is differentiated, this implies that $v_{1}=v_{2}$.
Now, every functional model satisfies these semantic conditions and every differentiated model that satisfies them is functional.
We thus obtain a completeness result for $\GPKF$ with respect to functional models:
\begin{corollary}
\label{soundAndCompleteFunctional}\label{soundnessAndCompletenessForSmilesatKF}
$\g\models_{\E_{\FUNC}}\varphi$ iff $\g\vdash_{\GPKF}\varphi$ for every $\g\cup\set{\varphi}\suq\setS$, where $\E_{\FUNC}$ is the class of functional models.
\end{corollary}

In contrast with what was the case for $\GPKT$, within such semantics for $\GPKF$  there are no longer countermodels  for (DM1.2$\ineg$)  or for (DM2.2$\uneg$).  At any rate, it should be clear that $\GPKF$ extends $\GPKD$, but does not extend $\GPKT$.

Going through quasi models we may prove cut-admissibility also for $\GPKF$.
However, unlike in previous cases, considering functional quasi models 
will not suffice. Indeed, there exist differentiated strengthened quasi models that respect $[\FUNC]$ whose accessibility relation is not a total function. 
Let a $\FUNC$~quasi model $\Q\M=\tup{\F,QV}$ based on a frame $\F=\tup{W,R}$ be a serial quasi model in which for every $w,v\in W$ such that $w R v$ we have, for every $\varphi\in\setS$, both
($\falQ{v}{\varphi}$ implies $\truQ{w}{\ineg\varphi}$) and ($\truQ{v}{\varphi}$ implies $\falQ{w}{\ineg\varphi}$).
We note that, although the accessibility relation in $\FUNC$~quasi models may not be a total function, we are still able to extract a functional model from it:

\begin{lemma}\label{func-instance}
	Every  $\FUNC$~quasi model has a functional instance. 
\end{lemma}
\begin{proof}
Let $\Q\M=\tup{\F,QV}$ be an $\FUNC$~quasi model based on a frame $\tup{W,R}$. Since $\Q\M$ is an $\FUNC$~quasi model, we have in particular that~$R$ is serial. Therefore, there exists some total function $R':W\rightarrow W$ such that $R'\suq R$. Let $\F'=\tup{W,R'}$. We define an appropriate valuation  $V:W\times\setS\rightarrow\set{\fff,\ttt}$ as in Lemma \ref{quasinstance}, while disregarding {(M5)}, and using the following instead of {(M4)} and {(M7)}:
\smallskip

\noindent
\begin{tabular}{rl}
  {(M4')} & $\varphi=\ineg\psi$, and $\fal{R'(w)}{\psi}$\\
  {(M7')} & $\varphi=\wfrown\psi$, and $\fal{w}{\varphi}$ and $\fal{w}{\ineg\varphi}$\\
\end{tabular}\smallskip

\noindent
The proof then carries on in a similar fashion to the proof of Lemma \ref{quasinstance}.
\end{proof}

\begin{corollary}
	$\GPKF$ enjoys cut-admissibility and is $\wfr'$-analytic, where $\wfr'$ is the restriction of $\wfr$ to the $\uneg$-free fragment of $\setS$, with an additional clause according to which $\ineg\varphi\wfr \wsmile\varphi$.
\end{corollary}

We include a word about further developments which could not be included here for reasons of space.  It is easy to see that~$\ineg$ and~$\uneg$ may be defined using the customary presentation of the modal logic~$\K$ by $\ineg\varphi:={\sim}\Box\varphi$ and $\uneg\varphi:=\Box{\sim}\varphi$. When considering only functional frames (like in $\GPKF$), we get a translation to $\GKF$ --- the ordinary modal logic of functional Kripke models. For the $\wsmile\wfrown$-free fragment of this logic, we may apply the general reduction to SAT proposed in \cite{lahavZoharSatBased}, which in particular means that the derivability problem for it is in co-NP. 
We further note that if one dismisses ${[{\vee}{\Ra}]}$ from the proof system,
derivability can be decided in linear time, by producing SAT-instances that consist solely of
Horn clauses. Such `half-disjunction' was also suggested in the context of primal infon logic \cite{Beklemishev29052012}, to obtain a linear time decision procedure.%

\subsection{Symmetry}

Let $\GPKB$ be the system obtained from $\GPK$
by replacing ${[{\ineg}{\Ra}]}$ and ${[{\Ra}{\ineg}]}$  with the following rules: 
{\small
$$[\B_{1}] \quad \dfrac{\Gamma,\ineg\Gamma',\varphi\Ra\Delta,\uneg\Delta'}{\uneg\Delta,\Delta'\Ra\uneg\varphi,\ineg\Gamma,\Gamma'}
\qquad\qquad\qquad [\B_{2}] \quad \dfrac{\Gamma,\ineg\Gamma'\Ra\varphi,\Delta,\uneg\Delta'}{\uneg\Delta,\Delta',\ineg\varphi\Ra\ineg\Gamma,\Gamma'}$$}%

\noindent
These correspond to the following basic rules:
$\tup{ p_{1}\Ra\;;\;\pi_{3}}\rs\Ra\uneg p_{1}$ and  $\tup{\Ra p_{1}\;;\;\pi_{3}}\rs\ineg p_{1}\Ra$, 
for the context relation $\pi_{3}=\{\tup{ p_{1}\Ra\;;\;\Ra\ineg p_{1}}$,\! $\tup{\ineg p_{1}\Ra\;;\;\Ra p_{1}}$,\! $\tup{\Ra p_{1}\;;\;\uneg p_{1}\Ra}$,\! $\tup{\Ra\uneg p_{1}\;;\; p_{1}\Ra}\}$.
This relation satisfies the following property: $s\,\pi_{3}\,q$ iff $\overline{q}\,\pi_{3}\,\overline{s}$, where $(\overline{\Ra\varphi})$ denotes $(\varphi\Ra)$ and $(\overline{\varphi\Ra})$ denotes $(\Ra\varphi)$.
By Proposition 4.28 of \cite{lah:avr:Unified2013}, the semantic condition these rules impose on strengthened models is symmetry of the accessibility relation. 
In addition, every symmetric model respects these rules, as well as the context relation $\pi_{3}$. 
It follows that:
\begin{corollary}\label{soundnessAndCompletenessForSmilesatKB}
$\g\models_{\E_{\B}}\varphi$ iff $\g\vdash_{\GPKB}\varphi$ for every $\g\cup\set{\varphi}\suq\setS$, 
where $\E_{\B}$ is the class of symmetric models. 
\end{corollary}

\noindent 
Symmetric frames are relevant from the viewpoint of sub-classical properties of negation.  They validate, for instance, the consecutions $\ineg\ineg p\models p$ and {$p\models \uneg\uneg p$}.
Paraconsistent logics based on symmetric (and reflexive) frames are also studied in~\cite{avr:zam:AiML16}, a paper that investigates in detail a conservative extension of the corresponding logic, obtained by the addition of a classical implication (but without primitive~$\uneg$ and~$\wfrown$), and offers for this logic a sequent system for which cut is not eliminable. 

Quasi models for $\GPKB$ are not necessarily symmetric, making it harder to convert them into instances in the form of symmetric models. This is why cut-admissibility for our system $\GPKB$ is here left open as a matter for further research.
However, using a similar technique of basic systems, it can be straightforwardly shown that $\GPKB$ is $\wfr$-analytic.
This does not require quasi models at all: one only has to show that every {\em partial} model, whose valuation's domain is closed under $\wfr$-subformulas, 
may be extended to a full model (see  Corollary 5.44 in~\cite{lah:avr:Unified2013}).

\section{Definability of classical negation}
\label{sec:definability}
 In this section we investigate definability of classical negation in the modal logics studied in this paper.
Given a set~$C$ of connectives and a logic~${\bf L}$, we denote by $\Lano{{\bf L}}{C}$ the $C$-free fragment of~${\bf L}$, that is, the restriction of~${\bf L}$ to the language without the connectives in~$C$.

\begin{theorem}~
\begin{enumerate}
	\item Classical negation is definable in the logics: 
	$\Lano{PKT}{\set{{\ineg},\wsmile}}$, $\Lano{PKT}{\set{{\uneg},\wfrown}}$, 
	$PKD$, and $PKF$.
	\item Classical negation is not definable in the logics: 
	 $PK$, $PKB$, $\Lano{PKT}{\set{\wsmile,\wfrown}}$, $\Lano{PKD}{\set{\wsmile}}$, $\Lano{PKD}{\set{\wfrown}}$, $\Lano{PKF}{\set{\wsmile}}$, and $\Lano{PKF}{\set{\wfrown}}$.
\end{enumerate}
\end{theorem}

\begin{proof}\\
(i) 
For $\Lano{PKT}{\set{{\ineg},\wsmile}}$ we set ${\sim}\varphi:=\uneg\varphi\vee\wfrown\varphi$, for $\Lano{PKT}{\set{{\uneg},\wfrown}}$ we set ${\sim}\varphi:=\ineg\varphi\wedge\wsmile\varphi$, and for $PKD$ and $PKF$ we set ${\sim}\varphi:=(\uneg\varphi\wedge\wsmile\varphi)\vee\wfrown\varphi$. 
It is easy to see that $\Ra\varphi,{\sim}\varphi$ and $\varphi,{\sim}\varphi\Ra$ are derivable in each system
for the defined connective~${\sim}$. 
Using cut, one obtains the usual sequent rules for classical negation.
\Cref{PKDDef} provides the derivations for $PKD$.
(Given that $\GPKF$ is a deductive extension of $\GPKD$, the derivation in \Cref{PKDDef} is also good for $PKF$.)

%



\begin{table}[hbt]
\centering
  \begin{tabular}{|c|}
    \hline\vspace{-2mm}\\
$\dfrac{
	\dfrac{
		\dfrac{
			\ddfrac{}{\varphi\Ra\varphi}~~~~
			\dfrac{
				\varphi\Ra\varphi
			}{
				\uneg\varphi\Ra\ineg\varphi
			}[\D]
			}{
				\varphi,\uneg\varphi,\wsmile\varphi\Ra
			}[{\wsmile}{\Ra}]
		}{
			\varphi,\uneg\varphi\wedge\wsmile\varphi\Ra
		}[{\land}{\Ra}]~~~~
		\dfrac{\varphi\Ra\varphi,\uneg\varphi}{
			\varphi,\wfrown\varphi\Ra
		}[{\wfrown}{\Ra}]
	}{
		\varphi,(\uneg\varphi\wedge\wsmile\varphi)\vee\wfrown\varphi\Ra
	}[{\lor}{\Ra}]$\\\\\hline\\
    $\dfrac{\dfrac{\dfrac{\varphi\Ra\varphi~~~~\uneg\varphi\Ra\uneg\varphi}{\Ra\varphi,\uneg\varphi,\wfrown\varphi}[{\Ra}{\wfrown}]~~~~\dfrac{\varphi,\ineg\varphi\Ra\varphi,\wfrown\varphi}{\Ra\varphi,\wsmile\varphi,\wfrown\varphi}[{\Ra}{\wsmile}]}{\Ra\varphi,(\uneg\varphi\wedge\wsmile\varphi),\wfrown\varphi}[{\Ra}{\land}]}{\Ra\varphi,(\uneg\varphi\wedge\wsmile\varphi)\vee\wfrown\varphi}[{\Ra}{\lor}]$\vspace{-2mm}\\\\
    \hline
      
  \end{tabular}
  \caption{
  }
  \label{PKDDef}
  \vspace{-4mm}
\end{table}

\medskip

\noindent
(ii) Let $X\in\{\!PK,$\! $PKB,$\! $\Lano{PKT}{\set{\wsmile,\wfrown}},$\! $\Lano{PKD}{\set{\wsmile}\!},$\! $\Lano{PKD}{\set{\wfrown}},$\! $\Lano{PKF}{\set{\wsmile}},$\! $\Lano{PKF}{\set{\wfrown}}\}$.
		Suppose for the sake of contradiction that classical negation~$\sim$ is definable in~$X$. Let $p\in\setP$ and~let $\varphi$ be ${\sim}(p)$. Then both $\Ra\varphi,p$ and $p,\varphi\Ra$ are valid in~$X$. 
		Consider a set $W$ that consists of two worlds, $w$ and $v$, and a valuation $V$ such that $V(w,q)=1$ and $V(v,q)=0$ for every atomic formula~$q$ (includ\-ing~$p$).
		Now, for each relation $R_X$ on $W$,  consider the model $\M_X=\tup{\tup{W,R_X},V}$.
		If $\M_X$ belongs to the class of models that semantically characterize $X$, then we must have  that $\M_X,w\Vdash\varphi,p\Ra$ and $\M_X,v\Vdash\;\Ra p,\varphi$. Since in $\M_X$ we have $\tru{w}{p}$ and $\fal{v}{p}$, we must then have $\fal{w}{\varphi}$ and $\tru{v}{\varphi}$. 
		We show that this is impossible, by structural induction on $\varphi$. More precisely, we claim that if $\fal{w}{\varphi}$ then $\fal{v}{\varphi}$.
		To show this, we consider the possible values for $X$, and define the accessibility relation $R_X$ in each case. 
		For $X\in\{PK,PKB\}$ define $R_{X}=\varnothing$, for $X=\Lano{PKT}{\set{\wsmile,\wfrown}}$ define $R_{X}=W\times W$, for $X\in\set{\Lano{PKD}{\set{\wsmile}},\Lano{PKF}{\set{\wsmile}}}$ define $R_{X}=\set{\tup{w,v},\tup{v,v}}$, and for $X\in\set{\Lano{PKD}{\set{\wfrown}},\Lano{PKF}{\set{\wfrown}}}$ define $R_{X}=\set{\tup{w,w},\tup{v,w}}$.
		We describe in detail only the third case.
			For this case, 
			note that since $R_{X}$ is a total function, 
			$\M_X$ belongs to the appropriate class of models, and $\ineg$ and $\uneg$ are indistinguishable, hence we may choose to consider~$\uneg$ instead of~$\ineg$. The cases where $\varphi$ is atomic, a conjunction, or a disjunction are trivial. If $\varphi=\uneg\psi$  for some~$\psi$ and $\fal{w}{\varphi}$, then we must have $\tru{v}{\psi}$ by [$\TT\uneg$], which implies by [$\FF\uneg$] that $\fal{v}{\varphi}$. If $\varphi=\wfrown\psi$ for some~$\psi$, then $\fal{v}{\varphi}$ must hold good: indeed, if on the one hand $\tru{w}{\uneg\psi}$ then $\fal{v}{\psi}$ by $[\FF\uneg]$, and hence $\tru{v}{\uneg\psi}$ by $[\TT\uneg]$, which implies by [$\FF\wfrown$] that $\fal{v}{\varphi}$; if on the other hand $\fal{w}{\uneg\psi}$ then $\tru{v}{\psi}$ by [$\TT\uneg$], and hence again $\fal{v}{\varphi}$ follows by [$\FF\wfrown$].
		\qed
\renewcommand{\qed}{}
\end{proof}

\section{This is possibly not the end}
\label{sec:coda}
 In contrast to the usual `positive modalities' of normal modal logics, which are monotone with respect to the underlying notion of consequence, we have devoted this paper to antitone connectives known as `negative modalities' --- specifically, to full type box-minus and full-type diamond-minus connectives.


Be they monotone or antitone on each of their arguments, the connectives of normal modal logics are always congruential: they treat equivalent formulas as synonymous.  The phenomenon seems to be an exception rather than the rule if many-valued logics with non-classical negations are involved.
For instance, Kleene's 3-valued logic fails to be congruential, as $p\land\neg p$ is equivalent to $q\land\neg q$, but their respective negations, $\neg(p\land\neg p)$ and $\neg(q\land\neg q)$, are not equivalent.
Also, the earliest paraconsistent logic in the literature (cf.~\cite{jas:48}) fails to be congruential, in spite of having been defined in terms of a translation into a fragment of the modal logic $S5$, and this failure remained unknown for decades (cf.~\cite{jmar:04g}).  The same holds for the other early paraconsistent logics developed later on, containing extra `strong negations' that live in the vicinity of classical negation (cf.~\cite{Nel:NaSoCiCS:59,daC:comptes:63}).
Of course, there are important `non-exceptions': intuitionistic logic and other intermediate logics constitute congruential paracomplete logics.  For another example perhaps more to the point, consider the four-valued logic of FDE, whose semantics may be formulated having as truth-values $\{\mathbf{t},\mathbf{b},\mathbf{n},\mathbf{f}\}$, where $\{\mathbf{t},\mathbf{b}\}$ are designated, the transitive reflexive closure of the order~$\leq$ such that $\mathbf{f}\leq\mathbf{b},\mathbf{n}\leq\mathbf{t}$ may be used to define $\land$ and $\lor$, respectively, as its meet and its join, while $\neg\langle\mathbf{t},\mathbf{b},\mathbf{n},\mathbf{f}\rangle:=\langle\mathbf{f},\mathbf{b},\mathbf{n},\mathbf{t}\rangle$.  
It is not hard to see that this logic is congruential and by defining the operators 
$\wsmile\langle\mathbf{t},\mathbf{b},\mathbf{n},\mathbf{f}\rangle:=\langle\mathbf{t},\mathbf{n},\mathbf{b},\mathbf{t}\rangle$ and 
$\wfrown\langle\mathbf{t},\mathbf{b},\mathbf{n},\mathbf{f}\rangle:=\langle\mathbf{f},\mathbf{n},\mathbf{b},\mathbf{f}\rangle$
it gets conservatively extended into another congruential logic that deductively extends our logic $PKF$ (but does not deductively extend $PKT$), if we interpret~$\ineg$ as~$\neg$.  It is worth noting that the latter logic is equivalent to the expansion of FDE by the addition of a classical negation.

Some terminological conventions and some concepts used in the present paper were borrowed or adapted from other fonts, sometimes without explicit reference.  For instance, in Section~\ref{sec:intro}, dadaistic and nihilistic models come from~\cite{jmarcos:neNMLiP}, and that paper also introduces the connectives~$\wsmile$ and~$\wfrown$ of the so-called Logics of Formal Inconsistency (cf.~\cite{car:jmar:Taxonomy}) and the dual Logics of Formal Undeterminedness (cf.~\cite{jmarcos:neNMLiP}, where the adjustment connectives are called connectives `of perfection').  The minimal conditions on negation, called $\llbracket$\textit{falsificatio}$\rrbracket$ and $\llbracket$\textit{verificatio}$\rrbracket$, come from~\cite{Mar:OnPLR}.
What we in the present paper call `determinacy' has in~\cite{dod:mar:ENTCS2013} been called `determinedness'.
The `strengthened models' from Section~\ref{sec:proofsystem} correspond to models with strongly-legal valuations in the terminology of~\cite{lah:avr:Unified2013}.
In Section~\ref{sec:extensions}, Rule [\D] may be thought of as a variation on the following well-known sequent rule for the modal logic $KD$: ${\g\Ra}\;/\;{\Box\g\Ra}$, and rules for \GPKT\ are variations on the usual sequent rule for the modal logic $KT$: ${\g,\varphi\Ra\d}\;/\;{\g,\Box\varphi\Ra\d}$ (cf.~\cite{Wansing}).
Also, the rule for \GPKF\ is a variation on the sequent rule from~\cite{kawai1987sequential} for the `Next' operator in the temporal logic $LTL$, namely: ${\g\Ra\d}\;/\;{\Box\g\Ra\Box\d}$.
We have not been able to find in the literature the obvious rules 
${\g,\varphi\Ra\d}\;/\;{\Box\g,\lozenge\varphi\Ra\lozenge\d}$ and 
${\g\Ra\varphi,\d}\;/\;{\Box\g\Ra\Box\varphi,\lozenge\d}$
for the modal logic~$K$ of which our rules $[{\ineg}{\Ra}]$  and $[{\Ra}{\uneg}]$ from Section~\ref{sec:proofsystem} would be variations~on. 
In Section~\ref{sec:analyticity}, the trick behind using three-valued models for addressing the admissibility of the cut rule goes at least as far back as \cite{Schutte60}.

The main feature of our approach here has been to rely on theoretical technology built elsewhere and show how it may be adapted to the present case.  Our hope is that this should prove a beneficial methodology, and that the idea of obtaining completeness and cut admissibility as particular applications of more general results will become more common, rather than proceeding always through \textit{ad hoc} completeness and cut elimination theorems.


While we have directed our attention, in this paper, to classes of frames that turned out to be particularly significative from the viewpoint of the relation between negative modalities of different types, we see two very natural ways of extending such study.
The \textit{first} natural extension would be to look at other classes of frames that prove to be relevant from the viewpoint of sub-classical properties of negation.  
For instance, 
it is easy to see that the class of frames with the Church-Rosser property validates $\ineg\ineg p\models \uneg\uneg p$, pinpointing an interesting consecution involving the interaction between negations of different types.    
Some other classes of frames deserving study do not seem to show the same amount of promise, from the viewpoint of paraconsistency or paracompleteness.  For instance, euclidean frames validate {$\llbracket\ineg$-explosion$\rrbracket$} 
if in the set of formulas $\{\ineg p,p\}$ one replaces~$p$ by $\ineg r$, and validate {$\llbracket\uneg$-implosion$\rrbracket$} if in $\{\uneg p,p\}$ one replaces~$p$ by $\uneg r$; also, transitive frames cause a similar behavior, but now swapping the roles of~$\ineg r$ and~$\uneg r$ in replacing~$p$.
Alternatively, a \textit{second} avenue worth exploring would lead us into logics containing more than one negative modality of the same type (as it has been done for logics with multiple paracomplete negations in~\cite{res:combpossneg:SL97}).
One could for instance consider not only the `forward-looking' negative modalities defined by the semantic clauses {[S$\ineg$]} and {[S$\uneg$]}, but also `backward-looking' negative modalities~$\ineg^{\!-1}$ and~$\uneg^{\!-1}$ defined by the clauses obtained from the latter ones by replacing $wRv$ by $vRw$ (such `converse modalities' have been studied in the context of temporal logic \cite{Prior67}, as well as in the context of the so-called Heyting-Brouwer logic \cite{rau:BH:DM1980}).
The interaction between the various negations would then be witnessed, in such extended language, by the validity over arbitrary frames of `pure' consecutions such as $\ineg^{\!-1}\ineg p\models p$ and $\ineg\ineg^{\!-1} p\models p$ (as well as $p\models \uneg^{\!-1}\uneg p$ and $p\models \uneg\uneg^{\!-1} p$), and the validity over symmetric frames of `mixed' consecutions such as $\uneg^{\!-1}\ineg p\models p$ and $\uneg\ineg^{\!-1} p\models p$ (as well as $p\models \ineg^{\!-1}\uneg p$ and $p\models \ineg\uneg^{\!-1} p$).
In our view, it seems worth the effort applying the machinery employed in the present paper to the above mentioned systems, and still others, in order to investigate results analogous to the ones we have here looked at.%
\footnote{The authors acknowledge partial support by the Marie Curie project GeTFun (PIRSES-GA-2012-318986) funded by EU-FP7, by CNPq and by The Israel Science Foundation (grant no.\ 817-15). They also take the chance to thank Hudson Benevides and three anonymous referees for the careful reading of an earlier version of this manuscript.}
%



\bibliographystyle{aiml16}
\bibliography{aiml16}

\begin{thebibliography}{10}
\expandafter\ifx\csname url\endcsname\relax
  \def\url#1{\texttt{#1}}\fi
\expandafter\ifx\csname urlprefix\endcsname\relax\def\urlprefix{URL }\fi
\newcommand{\enquote}[1]{``#1''}

\bibitem{avr:zam:AiML16}
Avron, A. and A.~Zamansky, \emph{A paraconsistent view on {$B$} and {$S5$}},
  this volume.

\bibitem{Beklemishev29052012}
Beklemishev, L. and Y.~Gurevich, \emph{Propositional primal logic with
  disjunction}, Journal of Logic and Computation \textbf{24} (2012),
  pp.~257--282.

\bibitem{car:jmar:Taxonomy}
{Carnielli}, W.~A. and J.~{Marcos}, \emph{A taxonomy of \textbf{C}-systems},
  in: W.~A. {Carnielli}, M.~E. {Coniglio} and I.~M.~L. {D'Ottaviano}, editors,
  \emph{{P}araconsistency: {T}he logical way to the inconsistent},  Lecture
  Notes in Pure and Applied Mathematics  \textbf{228}, Marcel Dekker, 2002 pp.
  1--94.

\bibitem{daC:comptes:63}
{da Costa}, N. C.~A., \emph{Calculs propositionnels pour le syst\`emes formels
  inconsistants}, Comptes Rendus Hebdomadaires des S\'{e}ances de
  l'Aca\-d\'e\-mie des Sciences\textup{, S\'{e}ries {A--B}} \textbf{257}
  (1963), pp.~3790--3793.

\bibitem{dod:mar:ENTCS2013}
Dod{\'o}, A. and J.~Marcos, \emph{Negative modalities, consistency and
  determinedness}, Electronic Notes in Theoretical Computer Science
  \textbf{300} (2014), pp.~21--45.

\bibitem{Dos:NMOiIL:1984}
{Do\v sen}, K., \emph{Negative modal operators in intuitionistic logic},
  Publications de {L'I}nsti\-tut {M}ath\'ematique (Beograd) (N.S.)
  \textbf{35(49)} (1984), pp.~3--14.

\bibitem{dun:zho:neggag:05}
Dunn, J.~M. and C.~Zhou, \emph{Negation in the context of {G}aggle {T}heory},
  Studia Logica \textbf{80} (2005), pp.~235--264.

\bibitem{jas:48}
{Ja{\'{s}}kowski}, S., \emph{A propositional calculus for inconsistent
  deductive systems (in {P}olish)}, Studia Societatis Scientiarum
  Torunensis\textup{, Sectio A} \textbf{5} (1948), pp.~57--77, translated into
  {E}nglish in {\it Studia Logica}, 24:143--157, 1967, and in {\it Logic and
  Logical Philosophy}, 7:35--56, 1999.

\bibitem{kawai1987sequential}
Kawai, H., \emph{Sequential calculus for a first order infinitary temporal
  logic}, Mathematical Logic Quarterly \textbf{33} (1987), pp.~423--432.

\bibitem{lah:avr:Unified2013}
Lahav, O. and A.~Avron, \emph{A unified semantic framework for fully structural
  propositional sequent systems}, ACM Transactions on Computational Logic
  \textbf{14} (2013), pp.~27:1--27:33.

\bibitem{lahavZoharSatBased}
Lahav, O. and Y.~Zohar, \emph{{SAT}-based decision procedure for analytic pure
  sequent calculi}, in: S.~Demri, D.~Kapur and C.~Weidenbach, editors,
  \emph{Automated Reasoning},  Lecture Notes in Computer Science
  \textbf{8562}, Springer International Publishing, 2014 pp. 76--90.

\bibitem{jmar:04g}
Marcos, J., \emph{Modality and paraconsistency}, in: M.~Bilkova and
  L.~Behounek, editors, \emph{The Logica Yearbook 2004}, Filosofia, 2005 pp.
  213--222.

\bibitem{jmarcos:neNMLiP}
Marcos, J., \emph{Nearly every normal modal logic is paranormal}, Logique et
  Analyse (N.S.) \textbf{48} (2005), pp.~279--300.

\bibitem{Mar:OnPLR}
{Marcos}, J., \emph{On negation: {P}ure local rules}, Journal of Applied Logic
  \textbf{3} (2005), pp.~185--219.

\bibitem{Nel:NaSoCiCS:59}
{Nelson}, D., \emph{Negation and separation of concepts in constructive
  systems}, in: A.~Heyting, editor, \emph{Constructivity in Mathematics},
  Studies in Logic and the Foundations of Mathematics, North-Holland,
  Amsterdam, 1959 pp. 208--225.

\bibitem{Prior67}
{Prior}, A., \enquote{Past, Present and Future,} Oxford University Press, 1967.

\bibitem{rau:BH:DM1980}
Rauszer, C., \emph{An algebraic and {K}ripke-style approach to a certain
  extension of {I}ntuitionistic {L}ogic}, Dissertationes Mathematicae
  \textbf{167} (1980).

\bibitem{res:combpossneg:SL97}
Restall, G., \emph{Combining possibilities and negations}, Studia Logica
  \textbf{59} (1997), pp.~121--141.

\bibitem{ripley:PhD}
Ripley, D.~W., \enquote{Negation in Natural Language,} Ph.D. thesis, University
  of North Carolina at Chapel Hill (2009).

\bibitem{Schutte60}
{Sch{\"{u}}tte}, K., \enquote{Beweistheorie,} Springer-Verlag, Berlin, 1960.

\bibitem{vaka:cons89:full}
{Vakarelov}, D., \emph{Consistency, completeness and negation}, in: G.~Priest,
  R.~Sylvan and J.~Norman, editors, \emph{Paraconsistent Logic: Essays on the
  inconsistent}, Philosophia Verlag, 1989 pp. 328--363.

\bibitem{Wansing}
{Wansing}, H., \emph{Sequent systems for modal logics}, in: D.~M. Gabbay and
  F.~Guenthner, editors, \emph{Handbook of Philosophical Logic}, Springer,
  2002, 2nd edition pp. 61--145, vol. 8.

\bibitem{Wojcicki88}
{W\'ojcicki}, R., \enquote{Theory of Logical Calculi,} Kluwer, Dordrecht, 1988.

\end{thebibliography}
\end{document}